\documentclass[a4paper,11pt,leqno]{article}
\usepackage{a4wide,color}
\usepackage{cite}
\usepackage{amsfonts,delarray,amssymb,amsmath,amsthm,a4,a4wide,color}
\usepackage{amsmath, euscript,xcolor,color,epsfig,graphics,dsfont}
\usepackage[english]{babel}
\usepackage[applemac]{inputenc}
\usepackage[cyr]{aeguill}
\usepackage{tikz}
\usepackage{hyperref}
\usepackage{mathtools}

\newtheorem{theo}{Theorem}
\newtheorem{pro}{Proposition}[section]
\newtheorem{lem}[pro]{Lemma}
\newtheorem{coro}[pro]{Corollary}
\newtheorem{remark}[pro]{Remark}
\newtheorem{defi}[pro]{Definition}

\usepackage{ esint }
\theoremstyle{plain}
\newtheorem*{theostar}{Theorem}

\def\({\left(}
\def\){\right)}
\def\1{\mathbf{1}}
\def\admissible{{\mathcal{A}_m}}

\def\div{\mathrm{div} \ }

\def\dt0{{{\frac{d}{dt}}_{|t=0}}}

\def\E{{\Sigma}}

\def\ep{\varepsilon}

\def\hal{\frac{1}{2}}

\def\I{{\mathcal{F}}}

\def\indic{\mathbf{1}}
\def\indic{\mathbf{1}}

\def\l|{\left|}

\def\mr{\mathbb{R}}
\def\mo{{\mu_0}}

\def\mz{\mathbb{Z}}
\def\nab{\nabla}

\def\p{\partial}

\def\ro{\rho}
\def\r|{\right|}

\def\supp{\text{Supp}}

\def\T{{\mathbb{T}}}

\def\vp{\varphi}

\def\probas{\mathcal{P}}
\def\dr{{\delta_\mr}}

\def\bw{{\underline{m}}}

\def\j{{E}}

\def\tW{\widetilde{W}}
\def\bw{\overline{W}}

\def\Esp{\mathbf{E}}
\def\lc{\left\langle}
\def\rc{\right\rangle}
\DeclarePairedDelimiter\floor{\lfloor}{\rfloor}

\renewcommand{\epsilon}{\varepsilon}
\def \N{\mathcal{N}}

\begin{document}

\title{A Uniqueness Result for Minimizers of the 1D Log-gas Renormalized Energy}
\author{Thomas Lebl\'e \footnote{Sorbonne Universités, UPMC Univ. Paris 06 and CNRS, UMR 7598, Laboratoire Jacques-Louis Lions, F-75005, Paris. E-mail: leble@ann.jussieu.fr}}
\maketitle
\begin{abstract}
In \cite{SSlg} Sandier and Serfaty studied the one-dimensional Log-gas model, in particular they gave a crystallization result by showing that the one-dimensional lattice $\mz$ is a minimizer for the so-called renormalized energy which they obtained as a limit of the $N$-particle Log-gas Hamiltonian for $N \to \infty$. However, this minimizer is not unique among infinite point configurations (for example small perturbations of $\mz$ leave the renormalized energy unchanged). In this paper, we establish that uniqueness holds at the level of (stationary) point processes, the only minimizer being given by averaging $\mz$ over a choice of the origin in $[0,1]$.  This is proved by showing a quantitative estimate on the two-point correlation function of a process in terms of its renormalized energy.   
 \end{abstract}

\section{Introduction and statement of the results}
\subsection{Introduction}
The $N$-particle Log-gas Hamiltonian $w_N$ is defined on $\mr^N$ by: 
\begin{equation}
\label{wn}
w_N (x_1, \dots, x_N)= -  \sum_{i \neq j} \log |x_i-x_j| +N  \sum_{i=1}^N V(x_i).
\end{equation}
where $V$ is a confining potential satisfying some growth conditions to be given later. 

While $w_N$ has an obvious physical interpretation as the  energy of $N$ particles $x_1, \dots, x_N$ living on the real line, interacting pairwise through a potential $g(x,y) = -\log|x-y|$ and subject to an external field $V$, the Hamiltonian (\ref{wn}) also appears in random matrix theory as an exponential weight in the law of the eigenvalues of random Hermitian matrices (for a survey see \cite{forrester}). Minimizers of $w_N$ are also called ($N$-points) “weighted Fekete sets” and arise in interpolation, cf. \cite{SaffTotik}. The “renormalized energy” $W$ of Sandier-Serfaty (introduced in \cite{SSvortex}, see also \cite{SS2d} for the two-dimensional case, \cite{RS} for an alternative approach that allows to handle the higher dimensional case as well, and  \cite{SZurich} for a pedagogical survey) is a way to make sense of the Hamiltonian $w_N$ in the $N \to \infty$ limit, by deriving an energy functional which allows to consider the energy of infinite point configurations, and which is the second-order $\Gamma$-limit of $w_N$. 

It is proven in \cite{SSlg} that $W$ is minimal at $\mz$ among infinite point configurations of density one, however this minimizer is not unique: in particular it was observed that local defects in the lattice, by the mean of arbitrary perturbations of $\mz$ on a compact set, form non-lattice configurations with the same minimal energy. In this paper we prove that the local defects essentially account for all the ground state degeneracy, by showing that there is no \textit{translation-invariant} probability measure on point configurations having minimal energy in expectation, but the one associated to $\mz$ by averaging $\mz$ over translations in $[0,1]$. This uniqueness result is obtained as consequence of our main theorem, which gives a lower bound on the (mean) renormalized energy of a stationary point process in terms of the difference between its two-point correlation function and that of the stationary process associated to the one-dimensional lattice $\mz$.

%Let us stress that this note is intended as a companion paper to \cite{SSlg}, whose notations we have tried to keep up with. As such the next two parts of our introduction strongly overlap with the introduction of \cite{SSlg}, and the proof itself is not self-contained in that we quote various results from the analysis of log-gases developed in \cite{SSlg}.

\subsection{Definition and properties of the renormalized energy}
In this section, and in all the paper, we follow mainly the definitions and notation from \cite{SSlg}. 

Let us start by recalling the definition of the renormalized energy.
%\vspace{-0.3cm}
%\paragraph{Renormalized energy of vector fields and point configurations}
The renormalized energy of an infinite configuration of points can be understood as a way of computing the electrostatic energy of those points, seen as interacting charged particles of charge $+1$, together with an infinite negatively charged uniform background. In 1D, the renormalized energy is obtained by “embedding” the real line into the plane and computing the renormalized energy in the plane according to its two-dimensional definition of  \cite{SS2d}. In particular, the pairwise interaction  $g(x,y) = -\log |x-y|$ is not the Coulomb electrostatic interaction of one-dimensional physics, but a restriction on $\mr \subset \mr^2$ of the Coulomb two-dimensional interaction, hence the term “Log-gases”.

In what follows, $\mr$ will denote the set of real numbers but also the real line of the plane $\mr^2$ i.e. points of the form $(x,0)\in \mr^2$.   For the sake of clarity, we will denote points in $\mr$ by the letter $x$ and points in the plane by $z=(x,y)$. We denote by
$\dr$ the measure of length on  $\mr$ seen as embedded in $\mr^2$, that is
$$\int_{\mr^2} \varphi  \dr  = \int_\mr \varphi(x, 0)\, dx$$
for any smooth compactly supported test function $\vp $ in $\mr^2$.

The “admissible classes” $\mathcal{A}_m$ correspond to the electric fields generated by infinite configurations on the real line together with a background of uniform density $m$:
\begin{defi} Let $m \geq 0$.
Let  $\j$ be a gradient vector field in $\mr^2$. We say $\j$ belongs to the admissible class $\mathcal{A}_m $
 if 
 \begin{equation} \label{eqj}
 \div \j= 2\pi (\nu -m\delta_\mr ) \quad  \text{in} \ \mr^2
 \end{equation}
 where $\nu$ has the form
$\nu=  \sum_{x \in \Lambda} \delta_{x}$ for
some discrete set $\Lambda \subset\mr \subset \mr^2$ (where $\delta_x$ denotes the Dirac mass at $x$), and $\frac{\nu ([-R,R] ) } {R}$ is bounded by a constant independent of $R>1$.
\end{defi}

We now turn to the central definition:
\begin{defi}[Renormalized energy]\label{def1}Let $m$ be a nonnegative number. For any bounded function $\chi$ and any
 $\j$ satisfying \eqref{eqj}
we let
\begin{equation}W(\j, \chi) = \lim_{\eta\to 0} \(
\hal\int_{\mr^2 \backslash \cup_{p\in\Lambda} B(p,\eta) }\chi
|\j|^2 +  \pi \log \eta \sum_{p\in\Lambda} \chi (p) \)
\end{equation}
and the renormalized energy $W$ is defined by
\begin{equation} \label{Wroi} W(\j)= \limsup_{R \to \infty}
\frac{W(\j, \chi_{R})}{R} ,
\end{equation} where $\{\chi_R\}_{R>0}$ is a family of cut-off functions satisfying 
$$|\nab \chi_{R}|\le C, \quad \supp(\chi_{R})
\subset [-R/2,R/2] \times \mr, \quad \chi_{R}(z)=1 \ \text{if } |x|<R/2-1 ,$$
for some $C$ independent of $R$. 
\end{defi}

The various admissible classes $\mathcal{A}_m$ ($m \in \mr^+$) are  related to each other by the following scaling relation: if $\j$ belongs to $\mathcal{A}_m$ then
$\j':=\frac{1}{m} \j(\cdot / m)$ belongs to $\mathcal{A}_1 $ and
\begin{equation}\label{scaling}
W(\j)= m \(W(\j') - \pi \log m\).
\end{equation}
Moreover, it is easy to see that the point configurations associated to $E$ and $E'$ coincide up to an homothety.

For reasons related to the physical interpretation of the Hamiltonian $w_N$, the gradient vector field $E$ is sometimes called the “electric field” associated to a configuration (seen as charged point particles). Starting from a discrete set of points $\Lambda \subset \mr \subset \mr^2$, there might be several gradient vector fields $E$ satisfying (\ref{eqj}) with $\nu = \sum_{x \in \Lambda} \delta_x$: if $E$ is any such field (let us note that, due to the infinite setting, there might not be any) we can simply add to $E$ the gradient of any harmonic function on $\mr^2$. In the two-dimensional case this is indeed an issue, but for Log-gases the following lemma shows that there is in fact a natural choice of the electric vector field $E$:

\begin{lem}{\cite[Lemma 1.7.]{SSlg}} \label{depoint} Let $\j\in \mathcal{A}_m$ be such that $W(\j)<+\infty$. Then any other $\j'$ satisfying $\div E' = \div E$ and $W(\j')<+\infty$, is such that $\j'=\j$. In other words, $W$ only depends on the points.
\end{lem}
By simple considerations similar to \cite[Section 1.2]{SS2d}  this makes $W$ a measurable function of the point configuration $\Lambda$ and with an abuse of notation we will write $W(\Lambda)$ as well as  $W(E)$, where $E$ is the only admissible vector field of finite energy associated to $\Lambda$ (when it exists). We will frequently use the following map to get from an electric field $E \in \mathcal{A}_1$ to its underlying point configuration:
\begin{equation} \label{pushforward}
E \mapsto \frac{1}{2\pi} \div E + \delta_{\mr}.
\end{equation}

It is not difficult to show (for a proof see \cite{SS2d}) that an admissible gradient vector field is in $L^q_{loc}(\mr^2, \mr^2)$ for $q < 2$. We endow the admissible classes $\mathcal{A}_m$ with the Borel $\sigma$-algebra inherited from $L^q_{loc}(\mr^2, \mr^2)$ for some $q < 2$.
\begin{defi}
Let $m>0$. If $P$ is a translation-invariant probability measure  on $\mathcal{A}_m$, we define 
\begin{equation}
\label{Wbar}
\bw(P) = \int W(E) dP(E).
\end{equation}We say that such a probability measure is translation-invariant (or stationary) when $P$ is invariant by (the push-forward of) the maps $E \mapsto E( \cdot-\lambda)$ for any $\lambda \in \mr$.
\end{defi}

Finally, when $X$ is a measurable space, $\mathcal{P}(X)$ will denote the space of probability measures on $X$. If $P\in \mathcal{P}(X)$ is a probability measure and $f: X \mapsto \mr$ a measurable function, we denote by $\Esp_P \left[ f \right]$ the expectation of $f$ under $P$.

\subsection{Periodic case and minimization}
When the configuration is assumed to have some periodicity, there is an explicit formula for $W$ in terms of the points. The following lemma is proven in \cite[Section 2.5.]{BS} (here we can reduce to the class $\mathcal{A}_1$ by scaling, as seen above in (\ref{scaling})).
\begin{lem}\label{casper}
In the case $m=1$ and when the set of points $\Lambda$ is periodic with respect to some lattice $N\mz$, then it can be viewed as a set of $N$ points $a_1, \dots , a_N$ over the torus
$\T_N := \mr/(N\mz)$.  In this case, by Lemma \ref{depoint} there exists a unique  $\j$ satisfying \eqref{eqj} and for which $W(\j)<+\infty$. It is periodic and equal to $\j_{\{a_i\}}= \nab H$, where $H$ is the  solution  on $\T_N$ to   $-\Delta H = 2\pi (\sum_i\delta_{a_i} - \dr)$, and we have  the explicit formula:
\begin{equation} \label{Wperiodique}
W(\j_{\{a_i\}})=
 -\frac{\pi}{N} \sum_{i \neq j} \log \left|2\sin \frac{\pi(a_i - a_j)}{N} \right|- \pi \log \frac{2\pi}{N}.
\end{equation}
\end{lem}
Henceforth we will denote by $W(\mz)$ the energy of the periodic electric field associated to $\mz$ as above. In this periodic, one-dimensional setting, $\mz$ is shown to be the (unique) minimizer of $W$ by a simple convexity argument. The key point of our proof is to make this argument quantitative in order to get a lower bound on $W(\j_{\{a_i\}})$ in terms of the local defects with respect to the lattice configuration (this is Lemma \ref{QLB}). 

A general argument of approximating any gradient vector field $E$ of finite energy by periodic electric fields implies a minimization result for $W$ on $\mathcal{A}_m$, without any periodicity assumption. It is proven in \cite[Theorem 2]{SSlg} that:
\begin{theostar}[crystallization in 1D] \label{thmini}
 $\min_{\admissible} W=  - \pi m \log (2\pi  m) $
and this minimum is achieved by the perfect lattice i.e. $\Lambda= \frac{1}{m} \mz$.
\end{theostar}
Let us emphasize that as a consequence of the definition of $W$ as a limit (\ref{Wroi}) over large intervals $W$ does not feel compact perturbations of the points (as long as the configuration stays simple i.e. two points of $\Lambda$ are always distinct) hence no uniqueness of the minimizer can be expected at the level of point configurations.

\subsection{Point processes and correlation functions}\label{ppp}
In this paragraph we give some definitions about point processes (for a complete presentation see \cite{PP1}).
\begin{defi}[Point processes]
Let $\mathcal{X}$ be the set of locally finite, simple point configurations on $\mr$. If $B \subset \mr$ is a Borel set, we let $N_{B} : \mathcal{X} \mapsto \mathbb{N}$ be the map giving the number of points of a configuration that lie in $B$. The set $\mathcal{X}$ is endowed with the initial $\sigma$-algebra associated to the maps $\{N_B, B \text{ Borel}\}$. 

A point process is a probability measure on $\mathcal{X}$. The additive group $\mr$ acts on $\mathcal{X}$ by translations $\{\theta_t\}_{t \in \mr}$: if $\mathcal{C} = \{x_i, i \in I\} \in \mathcal{X}$ we let $\theta_t \cdot \mathcal{C} = \{x_i - t, i \in I\}$. It also acts on the set $\probas(\mathcal{X})$ of point processes in the natural way, by pushing-forward $P \in \probas(\mathcal{X})$ by the map $\mathcal{C} \mapsto \theta_t \cdot \mathcal{C}$ for each $t \in \mr$.

A point process is said to be translation-invariant (or stationary) when it is invariant by the action of $\mr$.

It $\Lambda \in \mathcal{X}$ is a periodic configuration of points on $\mr$ with $\theta_L \cdot \Lambda = \Lambda$, we may associate to $\mathcal{C}$ the following stationary point process:
\begin{equation} \label{PLambda}  
P_{\Lambda} := \frac{1}{L} \int_{0}^L \delta_{\theta_{t} \cdot \Lambda} dt.
\end{equation}
\end{defi}
In particular, we will use the stationary processes associated to $\mz$ and its dilations $\frac{1}{m}\mz$ (for $m > 0$), which we denote by $P_{\mz}$, $P_{\frac{1}{m}\mz}$. We may abuse notation, relying on Lemma \ref{depoint}, and use the same notation for the stationary “electric” probability measure (concentrated on $\mathcal{A}_m$ and of finite energy) corresponding to $\frac{1}{m}\mz$.

\begin{defi}[Correlation functions] \label{correlationfunctions}
Let $P_{\Lambda} \in \probas(\mathcal{X})$ be a point process. For $k \geq 1$ the $k$-point correlation function $\rho_{k,P_{\Lambda}}$ is a linear form on the vector space of measurable functions $\varphi_{k} : \mr^k \longrightarrow \mr$, defined by:
$$\rho_{k, P_{\Lambda}}(\varphi_{k}) = \Esp_{P_\Lambda} \left( \sum_{x_1, \dots, x_k \in \mathcal{C} | x_i,x_j \text{ pairwise distinct}} \varphi_{k}(x_1, \dots, x_k)\right).$$
Strictly speaking, it is only defined on the subspace of functions $\varphi_k$ such that the map $\mathcal{C} \mapsto \sum_{x_1, \dots, x_k \in \mathcal{C} | x_i,x_j \text{ pairwise distinct}} \varphi_{k}(x_1, \dots, x_k)$ is integrable against $dP_{\Lambda}$.
\end{defi}
When the $k$-point correlation function exists as a distribution and can be identified with a measurable function, we will write $\int \rho_{k} \varphi_k$ instead of $\rho_{k} (\varphi_k)$. Heuristically, $\rho_1$ (also called the intensity of the point process) gives the density of the process at each point, while $\rho_2(x,y)$ gives the probability of having a point both at $x$ and $y$.

\begin{remark} \label{meandensitym}If $P$ is a translation-invariant (stationary) probability measure concentrated on the class of admissible electric fields $\mathcal{A}_m$, the push-forward $P_{\Lambda}$ of $P$ by $E \mapsto \frac{1}{2\pi}  \div E + m\delta_{\mr}$ is a stationary point process. Let us assume that $\bw(P)$ is finite. Then:
\begin{itemize}
\item The one-point correlation function may be identified as (testing against) the function $\rho_{1,P_{\Lambda}} \equiv m$.
\item  The two-point correlation $\rho_{2, P_{\Lambda}}$ is well-defined as a Radon measure on $\mr^2$.
\end{itemize}
\end{remark}
Physically speaking, this is because there must be approximately $m$ points per unit volume in order to compensate (without overwhelming) the background charge $-m\delta_{\mr}$, so that the configuration is globally neutral (non-neutrality would generate too much energy). We will give a proof of Remark \ref{meandensitym} in Section \ref{secdiscr}.
 
Henceforth, if $\varphi$ is a compactly supported, continous function on $\mr^2$ and $\mathcal{C}$ a point configuration, we will denote by $\lc \varphi, \mathcal{C} \rc$ the quantity
$$\lc \varphi, \mathcal{C} \rc := \sum_{x\neq y | x,y \in \mathcal{C}} \varphi(x,y)$$
which is always a well-defined number since $\varphi$ is compactly supported and $\mathcal{C} \in \mathcal{X}$ is locally finite.

\subsection{Statement of the results}
We start by stating a quantitative version of the convexity argument on which the minimality of $W(\mz)$ as in \cite[Proposition 4.3.]{SSlg} relies: 
\begin{lem} \label{QLB}
Let $a_1 < \cdots < a_N$ be any points in $[0,N]$ and $\j_{\{a_i\}}$ be the corresponding periodic vector field as in Lemma \ref{casper}. Let $u_{p,i}= a_{i+p}-a_{i}$ with the convention $a_{N+l}=a_l+N$ and let $b_{p,i} = u_{p,i} - p$. Then, for some universal constant $C$:
\begin{equation} \label{eqQLB}
W(\j_{\{a_i\}}) - W(\mathbb{Z}) \geq C \sum_{p=1}^{N/2} \frac{1}{N} \sum_{i=1}^N \min\left(\frac{b^2_{p,i}}{p^2}, 1\right).
\end{equation}
\end{lem} 
The proof is given in section \ref{preuvequantitative}. The quantity $b_{p,i}$ in the right-hand side of \eqref{eqQLB} measures a local defect with respect to the lattice: the spacing error $u_{p,i} - p$ between two $p$-th neighbours (in $\mz$ two $p$-th neighbours are always at distance $p$).

We then state our main theorem and its consequence for the minimization problem. The theorem gives for any translation-invariant probability measure $P \in \probas(\mathcal{A}_1)$ of finite energy a lower bound on $\bw(P) - W(\mz)$ in terms of the two-point correlation functions of the stationary point process associated to $P$.

\begin{theo} \label{NonPer} Let $P$ be a stationary probability measure concentrated on $\mathcal{A}_1$ such that $\bw(P)$ is finite, let $P_{\Lambda}$ be the push-forward of $P$ by the map (\ref{pushforward}) and $\rho_{2, P_{\Lambda}}$ be the two-point correlation function of $P_{\Lambda}$. Then for any function $\varphi \in C^{1}_c(\mr \times \mr)$ we have
\begin{equation}
\label{inegsurw}
\left| \int  (\ro_{2,{P_{\Lambda}}}  -\ro_{2, P_{\mz}}) \varphi  \right|\le C_{\varphi} (\bw(P)+C)^{\hal} \( \bw(P) - W(\mz) \)^{\hal} 
\end{equation}
with $C_{\varphi}$ depending only on $||\varphi||_{\infty}, ||\nabla\varphi||_{\infty}$ and on a $T \geq 1$ such that $\varphi$ is supported on $[-T, T]^2$, and $C$ a universal constant.
\end{theo}
The proof is given in Section \ref{preuveTheoreme}. An easy consequence of Theorem \ref{NonPer} is the following uniqueness result:
\begin{coro} \label{corollaire1}  For $m > 0$, the only minimizer of $\bw$ on the set of stationary probability measures concentrated on $\mathcal{A}_m$ is $P_{\frac{1}{m}\mz}$.
\end{coro}
Theorem \ref{NonPer} also allows to track down the crystallization of the statistical mechanics model via the convergence (in distributional sense) of the two-point correlation functions of  to $\rho_{2,\mz}$ when the inverse temperature $\beta \rightarrow + \infty$, as stated in \cite[Corollary 1.14.]{SSlg}. This was one of the main motivations for this paper.

\subsection{Connection with the Log-gas Hamiltonian}
\label{motivation}
Let us now briefly give a motivation for studying the renormalized energy and its minimization (for a thorough study we refer again to \cite{SS2d}, \cite{SSlg}, \cite{RS} and to the survey \cite{SZurich}). 

Starting again from $w_N$, when the potential $V$ in the definition \eqref{wn} is lower semi-continuous and satisfies the growth assumption $V(x) - 2 \log|x| \to + \infty$ (when $|x| \to + \infty$), it is known that for any sequence $\{\vec{x_N}\}_N$ of minimizers of $w_N$, the empirical measures $\mu_N := \frac{1}{N} \sum_{i=1}^N \delta_{x_i}$ converge weakly as $N \to +\infty$ to some measure $\mu_0$ on $\mr$, called the “equilibrium measure”. We assume that the equilibirum measure $\mu_0$ has a compact support $\Sigma$ which is a finite union of compact intervals, and is absolutely continous with respect to the Lebesgue measure with a density $m_0 \in L^{\infty}(\mr)$. The result is in fact much stronger, since the convergence of the empirical measures to $\mu_0$ holds not only for minimizers but almost surely under the Gibbs measure associated to $w_N$ at any finite temperature (see \cite{benarousguionnet},\cite{hiaipetz}). The renormalized energy appears as the Hamiltonian $w_N$ at second order: there is an exact splitting formula
$$w_N(x_1, \dots, x_N)=  N^2 \I (\mo)- N \log N + N F_N(\nu)$$
where $\I (\mo)$ is a “first-order” potential energy associated to the equilibrium measure, and $F_N$ is a function of the finite point configuration $\sum_{i=1}^N \delta_{x_i}$. To any such finite point configuration we associate a probability measure $P_{\nu_N}$ on $\Sigma \times \mathcal{A}$ obtained by averaging over $x \in \Sigma$ the electric field $E_N$ associated to the finite configuration $\nu'_N = \sum_{i=1}^N \delta_{N x_i}$ (for finite configurations such a field always exists), translated by $Nx$:
$$
P_{\nu_N} := \fint_{\Sigma} \delta_{\left(x, E_N(Nx+\cdot)\right)}.
$$
Let us emphasize that in this setting the scaling $x_i \mapsto x'_i = Nx_i$ is necessary: since we know that the empirical measures $\frac{1}{N}  \sum_{i=1}^N \delta_{x_i}$ typically converge to a compactly supported measure $\mu_0$ on $\mr$ it is relevant to scale the distances by $N$ so that the spacing between two consecutive points becomes of order $1$. If $\{\nu_N\}_N$ is a sequence of finite point configurations such that $\{F_N(\nu_N)\}_N$ is bounded then it is proven in \cite[Theorem 3.]{SSlg} that up to extraction the sequence $\{P_{\nu_N}\}_N$ converges to some admissible probability measure $P$ (the definition of “admissible” is given below) and $\liminf_{N \to + \infty}  F_N(\nu_N) \geq \tW(P)$. More precisely, the sequence of functionals $\{F_N\}_N$ (for each $N$ we can see $F_N$ as a function of probability measures $P \in \mathcal{P}(\Sigma \times \mathcal{A})$ which is infinite outside the image of the map $\nu_N \mapsto P_{\nu_N}$) has $\tW$ for $\Gamma$-limit (see \cite{braides}), which implies that the minimizers of $F_N$ (hence of $w_N$) converge to minimizers of $\tW$. This reduces the second-order study of the Hamiltonian $w_N$ in the limit $N \to \infty$ to the study of $\tW$ on admissible probabilities. 

\begin{defi}[Admissible probabilities] \label{admissible}
We say $P\in \mathcal{P}(\Sigma \times \mathcal{A})$ is \em{admissible} if 
\begin{itemize}
\item The first marginal of $P$ is the normalized Lebesgue measure on $\E$.
\item It holds for $P$-a.e. $(x, E)$ that $\j \in \mathcal{A}_{m_0(x)}$.
\item $P$ is invariant by the maps $(x,E) \mapsto (x, E(\lambda + \cdot))$ for all $\lambda \in \mr$ (this is a weaker assumption than the “$T_{\lambda(x)}$-invariance” of \cite{SSlg} but it is sufficient for our purpose).
\end{itemize} 
When $P$ is admissible, we denote by $\tW$ the expectation of $W$ under $P$:
\begin{equation}
\label{Wtilde}
\tW(P) = \frac{|\E|}{\pi} \int W(E)\, dP(x, E). 
\end{equation}
\end{defi}
For Log-gases, Theorem \ref{NonPer} implies the following uniqueness result:
\begin{coro} \label{corollaire2}  The equilibrium measure $\mu_0$ and its density $m_0$ being fixed, the only minimizer of $\tW$ on the set of admissible probability measures is given by \begin{equation} 
\label{P0}
 P_0=\frac{dx_{|\E}}{|\E|} \otimes  P_{\frac{1}{m_0(x)} \mathbb{Z}}.\end{equation}
\end{coro}
Let us note that $P_{\frac{1}{m}\mz}$ is not well-defined for $m = 0$, however we may assume that the set $\{x \in \Sigma, m_0(x) = 0\}$ has zero Lebesgue measure. Corollary \ref{corollaire2} as well as Corollary \ref{corollaire1} are proven in Section \ref{uniqueness}.

The minimization problem is of physical relevance, indeed minimizers of an Hamiltonian describe the behaviour of the physical system at zero temperature. It is believed (see e.g. \cite[Conjecture 9.4.]{CohnKumar}) that for a wide range of interactions the minimizing infinite configurations are lattices. The one-dimensional crystallization is somewhat easier to prove (see e.g. \cite{lenard1D} for positive results concerning the Coulomb interaction in $1d$, \cite{bl1}, \cite{bl2}, \cite{kunz} for one-dimensional long-range order, \cite{avilalastsimon} for zeroes of orthogonal polynomials), but the higher-dimensional cases are largely open (see \cite{Theil1} for a result in $d=2$ and \cite{Theil2} for recent progress in $d=3$). 
\section{Proof of the results}
\subsection{Preliminary bounds on the density of points} \label{secdiscr}
Since we are dealing with two-point correlation functions, we will often need to bound the variance (for some point process) of the number of points that lie in some fixed interval. For this purpose we use a deviation estimate of \cite{SSlg} which gives a lower bound on the renormalized energy in terms of the local non-neutrality of a point configuration (together with the uniform background). We summarize the consequences  for correlation functions in the following lemma:

\begin{lem}\label{discrepa} If $E \in \mathcal{A}_1$ is an electric field, we denote by 
$\mathcal{N}(E,T)$ the number of points of $\frac{1}{2\pi} \div E +  \delta_{\mr}$ that lie in the inverval $[-T, T]$. Let $P$ be a stationary probability measure concentrated on $\mathcal{A}_1$. 
\begin{itemize}
\item We have, for any $T > 0$,
$$\Esp_{P} \left[ \mathcal{N}(E,T) \right] = 2T.$$ In particular, the one-point correlation function (the intensity) of the point process associated to $P$ may be identified as (testing against) the function $\rho_{1,P_{\Lambda}} \equiv 1$. This, together with the scaling relation \eqref{scaling}, proves the first claim of Remark \ref{meandensitym}.
\item The following bound holds:
$$\int \mathcal{N}(E,T)^2 dP(E) \leq C_T \left(C + \bw(P)\right).$$
where $C_T$ depends only on $T$ and $C$ is  universal.
\end{itemize}
\end{lem}

\begin{proof} In what follows $C$ is a universal constant, which may vary from line to line. Let $E$ be in $\mathcal{A}_1$ and $T \geq 1$ and let $\Lambda(E) : = \frac{1}{2\pi} \div E + \delta_{\mr}$. Denoting by $D(E, T)$ the quantity $D(E,T) = \mathcal{N}(E,T) - 2T$ ($D$ measures the “discrepancy” between the expected number of points and the actual one, hence is a measure of local non-neutrality), \cite[Proposition 4.6.]{SSlg} reads:
$$ \int_{[-2T, 2T]}  dg \ge -CT + cD(E, T)^2 \min\(1, \frac{|D(E,T)|}{T}\)
$$
where the density $g$ is defined in \cite[Proposition 2.1.]{SSlg} (we quote below the results that we need) and $c > 0$ is universal. From \cite[Proposition 2.1.i)]{SSlg}, we know that $g$ is bounded below by $-C$, hence if $\chi_T$ is a smooth cut-off function satisfying $\chi_T \equiv 1$ on $[-2T, 2T]$ and $\chi_T \equiv 0$ outside $[-2(T+1), 2(T+1)]$ with $\|\nab\chi_T\|_\infty \leq 1$, which we extend by $\bar{\chi}_T(x,y) = \chi_T(x)$ on $\mr^2$, we have:
\begin{equation} \label{deviation2}
\int \bar{\chi}_T dg \ge -CT + cD(E, T)^2 \min\(1, \frac{|D(E,T)|}{T}\).
\end{equation}
We also know from \cite[Proposition 2.1.iii)]{SSlg} that for such a function $\chi_T$ the following bound holds: 
\begin{equation} \label{Wchi} 
\left|W(\j,\chi_T) - \int \bar{\chi}_T\,dg\right|
\le C n (\log n + 1)\|\nab\chi\|_\infty,
\end{equation}
where $n$ is a boundary term bounded by the number of points of the configuration $\Lambda(E)$ in $[-3T, -2T]~\cup~[2T, 3T]$. Combining (\ref{deviation2}) and (\ref{Wchi}) we easily get
$$
D(E, T)^2 \leq C \left(C + T + W(\j,\chi_T) + n (\log n + 1)\right).
$$
Taking the expectation under $P$ yields:
\begin{equation} \label{deviation3bis}
\int D(E)^2 dP(E) \leq C \left(C + T + \int \left(W(\j,\chi_T) + n \log n + n \right)dP(E)\right)
\end{equation}
with $n = n(E)$ bounded by the number of points of $\Lambda(E)$ in $[-3T, -2T] \cup [2T, 3T]$. Since $P$ is stationary the average number (under $P$) of points of $\Lambda(E)$ in any interval of length $T$ is the same. Hence if we write $\N_+(E) = \max(\N(E,T), 2)$ we have by stationarity (and using the Cauchy-Schwarz inequality):
\begin{equation}\label{deviation4}
\int n \log n dP \leq 4 \left(\int \N_+^2 dP\right)^{\hal} \left(\int (\log \N_+)^2 dP \right)^{\hal}.
\end{equation}
Obviously the right-hand side of \eqref{deviation4} also bounds the term $\int n(E) dP(E)$. For any $0 < \alpha \leq \hal$, we may find $C_{\alpha}$ large enough such that 
$$\left(\int \N_+^2 dP\right)^{\hal} \left(\int (\log \N_+)^2 dP \right)^{\hal} \leq C_{\alpha} \left(\int \N_+^2 dP\right)^{\hal + \alpha}.$$
Together with \eqref{deviation3bis} and \eqref{deviation4} we thus obtain:
\begin{equation} \label{deviationrou} \int D(E, T)^2 dP(E) \leq C \left(C + T + \int W(\j,\chi_T) dP(E) + C_{\alpha} \left(\int\N_+^2 dP\right)^{\hal + \alpha}  \right).
\end{equation}
Now by stationarity of $P$ and additivity of $W(E, \cdot)$ (see Definition \ref{def1}) we have for any $R > 0$:
$$\int W(\j,\chi_T) dP(E) = \frac{1}{2R} \int_{-R}^R \int W(\j, \chi_T(x + \cdot)) dP(E) dx = \frac{1}{2R} \int W(E, \mathbf{1}_{[-R,R]} \star \chi_T) dP(E).$$
Let us observe that the family of functions $\{\chi'_{R}\}_{R}$ defined as $$\chi'_R := \frac{\mathbf{1}_{[-R,R]} \star \chi_T}{\int \chi_T}$$ satisfy the conditions of Definition \ref{def1} so that sending $R$ to $+ \infty$ we get: $\int W(\j,\chi'_R) dP(E) = \bw(P)$. This in turn implies that
\begin{equation} \label{WjchiT}
\int W(\j,\chi_T) dP(E) = \bw(P) \left(\int \chi_T\right)  \leq CT(\bw(P) + C).
\end{equation}
Moreover since $D(E,T) = \mathcal{N}(E,T) - 2T$ an elementary computation shows that for any $0 < \alpha \leq \hal$ we have:
\begin{equation} \label{alphaalpha'}
CC_{\alpha} \left(\int\N_+^2 dP\right)^{\hal + \alpha} \leq \hal \int D^2(E,T) dP + C'_{\alpha}T^{1+\alpha}
\end{equation}
where $C, C_{\alpha}$ are the constants in \eqref{deviationrou} and $C'_{\alpha}$ depends only on $\alpha$.
Combining \eqref{deviationrou}, \eqref{WjchiT} and \eqref{alphaalpha'} we get for any $0 < \alpha \leq \hal$:
\begin{equation} \label{bornediscrepance}
\int D(E, T)^2 dP(E)  \leq CT(\bw(P) +C) + C'_{\alpha}T^{1 + 2\alpha}
\end{equation}
with a constant $C'_{\alpha}$ depending only on $\alpha$ and $C$ universal. Equation \eqref{bornediscrepance} implies the following:
\begin{itemize}
\item We have $\int (\N(E,T) - 2T)^2 dP(E) = o(T^2)$ hence $\int (\N(E,T) - 2T) dP(E) = o(T)$ but by stationarity we have $\int \N(E,T) dP(E) = T \int \N(E,1) dP(E)$ so that in fact $$\int \N(E,T) dP(E) = 2T$$ for all $T > 0$, which proves the first claim of the lemma.
\item Since $\N(E,T)^2 \leq 4T^2 + 2 \int D(E, T)^2$ taking $\alpha = \hal$ we also get a bound on the mean square number of points $\int \N(E,T)^2 dP(E)$ as in the second claim of the lemma.
\end{itemize}
\end{proof}

We now use Lemma \ref{discrepa} to show that two-point correlation functions exist as Radon measures for point processes of finite renormalized energy.
\begin{lem} \label{boundphiW} Let $\varphi \in C^{0}_c(\mr \times \mr)$ and $P$ be a stationary probability measure on $\mathcal{A}_1$ such that $\bw(P)$ is finite. Let also $P_{\Lambda}$ be the push-forward of $P$ by the map $E \mapsto \frac{1}{2\pi} \div E +  \delta_{\mr}$. The following bound holds:
\begin{equation}
\left| \int \lc \varphi, \mathcal{C} \rc dP_{\Lambda}(\mathcal{C})\right| \leq C_{\varphi} (\bw(P)+C)
\end{equation}
with $C_{\varphi}$ depending only on $||\varphi||_{\infty}$ and on a $T \geq 1$ such that $\varphi$ is supported on $[-T, T]^2$ and $C$ a universal constant.
\end{lem}
\begin{proof} 
Let $T \geq 1 $ such that $\varphi$ is supported on $[-T, T]^2$.
For $x \in \mr$ and $E \in \mathcal{A}_1$, let us denote by $\mathcal{N}(E, T)$ the number of points of $\frac{1}{2\pi} \div E +  \delta_{\mr}$ lying in $[-T, T]$. The following bound is obvious by definition
$$\left|\lc \varphi, \mathcal{C} \rc \right| \leq \mathcal{N}(E,T)^2 ||\varphi||_{\infty}.$$
Integrating against $dP_{\Lambda}$, we get
\begin{equation}
\label{bornephiE}
\left| \int \lc \varphi, \mathcal{C} \rc dP_{\Lambda}(\mathcal{C})\right| \leq ||\varphi||_{\infty} \int \mathcal{N}(E,T)^2 dP(E).
\end{equation}
By Lemma \ref{discrepa} we know that the right-hand side of \eqref{bornephiE} is bounded by $C_{\varphi} (C + \bw(P))$ which concludes the proof.
\end{proof}

Lemma \ref{boundphiW} has the following implication: if $P$ is a stationary probability measure concentrated on $\mathcal{A}_1$ such that $\bw(P)$ is finite, then the push-forward of $P$ by $E \mapsto \frac{1}{2\pi} \div E + \delta_{\mr}$ admits a two-point correlation function in distributional sense. Indeed, the linear form
$\varphi \mapsto \int \lc \varphi, E \rc dP(E)$ is shown to be bounded by $O(||\varphi||_{\infty})$ uniformly for test functions in $C^0([-T, T]^2)$, for all $T$. This proves the second claim of Remark \ref{meandensitym} (the case of probability measures concentrated on $\mathcal{A}_m$ reduces to the former case by scaling as in \eqref{scaling}).

\subsection{Energy lower bound near the ground state for periodic configurations}
\label{preuvequantitative}
We prove a quantitative version of the minimization of $W$ on periodic configurations, as stated in Lemma \ref{QLB}.
\begin{proof} 
Let $u_{p,i}= a_{i+p}-a_{i}$, with the convention $a_{N+l}=a_l+N$, and let $b_{p,i} = u_{p,i} - p$. We know from (\ref{Wperiodique}) that 
\begin{equation} \label{DeltaW}
W(\j_{\{a_i\}}) - W(\mathbb{Z}) = \frac{2\pi}{N} \sum_{p=1}^{N/2} \left( \log\left|2\sin \frac{p\pi}{N}\right| - \frac{1}{N} \sum_{i=1}^N \log\left|2 \sin\frac{\pi u_{p,i}}{N}\right| \right).
\end{equation}

Using a Taylor expansion of the function $F : x \mapsto \log|2 \sin x|$, we get for each $p,i$: 
\begin{multline}\label{Taylor}
F\left(\frac{\pi u_{p,i}}{N}\right)  = F\left(\frac{1}{N} \sum_{i=1}^N \frac{\pi u_{p,i} }{N}\right) + F'\left(\frac{1}{N}  \sum_{i=1}^N \frac{\pi u_{p,i} }{N}\right) \left(\frac{\pi u_{p,i}}{N} - \frac{1}{N} \sum_{i=1}^N \frac{\pi u_{p,i}}{N}\right) \\ + \frac{1}{2} F''(x_{p,i}) \left(\frac{\pi u_{p,i}}{N} - \frac{1}{N} \sum_{i=1}^N \frac{\pi u_{p,i}}{N}\right)^2
\end{multline}
for a certain $x_{p,i}$ with $$|x_{p,i}| \leq \max\left(\frac{\pi |u_{p,i}|}{N},  \frac{1}{N}\left| \sum_{i=1}^N \frac{\pi u_{p,i}}{N}\right|\right).$$ Observing that $\sum_{i=1}^N u_{p,i} = pN$, 
we have $|x_{p,i}| \leq \frac{p\pi}{N} + \frac{|b_{p,i}|\pi}{N}$ and 
\begin{equation} \label{minoxpi}
\frac{1}{x_{p,i}^2} \left(\frac{\pi u_{p,i}}{N} - \frac{\pi p}{N} \right)^2 = \frac{\pi^2b_{p,i}^2}{N^2 x_{p,i}^2}  \geq \frac{1}{2} \frac{\pi^2 b_{p,i}^2}{(p\pi)^2 + (b_{p,i}\pi)^2} \geq \frac{1}{6} \min\left(\frac{b^2_{p,i}}{p^2}, 1\right). 
\end{equation}
The last inequality in \eqref{minoxpi} is obtained by observing that $\frac{1}{2} \frac{x^2}{p^2+x^2} \geq \frac{1}{6} \min(\frac{x^2}{p^2},1)$ on $\mr$.
Summing the Taylor expansions (\ref{Taylor}) for $i=1 \dots N$ gives, for any $p \leq N/2$:
\begin{equation} \label{sumTaylor}
 \log \left|2\sin \frac{p\pi}{N}\right| - \frac{1}{N} \sum_{i=1}^N \log\left|2 \sin\frac{\pi u_{p,i}}{N}\right|   = \sum_{i=1}^N \frac{1}{2} F''(x_{p,i})   \left(\frac{\pi u_{p,i}}{N} - \frac{\pi p}{N} \right)^2.
\end{equation}
An explicit computation shows that, for any $x \in \mr$, $F''(x) = \frac{1}{\sin^2x} \geq \max\left(1, \frac{1}{x^2}\right)$, so by combining (\ref{minoxpi}) and (\ref{sumTaylor}) we get
$$\log \left|2\sin \frac{p\pi}{N}\right| - \frac{1}{N} \sum_{i=1}^N \log\left|2 \sin\frac{\pi u_{p,i}}{N}\right|   \geq \sum_{i=1}^N \frac{1}{2} \max\left(\frac{\pi^2b_{p,i}^2}{N^2}, \frac{1}{6} \min\left(\frac{b^2_{p,i}}{p^2}, 1\right) \right).$$
Finally, inserting the previous inequality for $1 \leq p \leq N/2$ into (\ref{DeltaW}) gives
\begin{equation} \label{DeltaW3}
W(\j_{\{a_i\}}) - W(\mathbb{Z}) \geq  \sum_{p=1}^{N/2} \frac{1}{2} \max\left(\frac{\pi^2b_{p,i}^2}{N^2}, \frac{1}{6} \min\left(\frac{b^2_{p,i}}{p^2}, 1\right) \right)
\end{equation}
which yields the inequality 
$$W(\j_{\{a_i\}}) - W(\mathbb{Z}) \geq C \sum_{p=1}^{N/2} \frac{1}{N} \sum_{i=1}^N \min\left(\frac{b^2_{p,i}}{p^2}, 1\right)$$
for some universal constant $C$.
\end{proof}

\subsection{Consequences for correlation functions}
We now recast Lemma \ref{QLB} in the context of stationary point processes associated to periodic point configurations:
\begin{lem}\label{LemmeDeltarho} For any $N \geq 1$, let $a_1 < \cdots < a_N$ be any points in $[0,N]$ and $\j_{\{a_i\}}$ be the corresponding periodic vector field. Let $\Lambda$ be the corresponding infinite periodic configuration in $\mr$, and $P_{\Lambda}$ be the stationary point process associated to $\Lambda$, defined in (\ref{PLambda}) by averaging translated copies of $\Lambda$ over $[0,N]$. Assume that $W(\j_{\{a_i\}})$ is finite. The following bound holds:
$$\left| \int \left(\rho_{2,P_{\mz}} - \rho_{2, P_{\Lambda}}\right) \varphi \right| \leq C_{\varphi} (C+W(\j_{\{a_i\}}))^{\hal} \left(W(\j_{\{a_i\}}) - W(\mz)\right)^{\frac{1}{2}}$$
for any $\varphi \in C^{1}_c(\mr \times \mr)$ with $C_{\varphi}$ depending only on $||\varphi||_{\infty}, ||\nabla \varphi||_{\infty}$ and on a $T \geq 1$ such that $\varphi$ is supported on $[-T, T]^2$, and $C$ universal.
\end{lem}
\begin{proof}
Let $\varphi \in C^{1}_c([-T, T]^2)$ (without loss of generality we assume $T \geq 1$). Since $W(\j_{\{a_i\}})$ is finite we know by Remark \ref{meandensitym} that $\rho_{2, P_{\Lambda}}$ exists as a Radon measure, we will abuse notation and write $\int \rho_{2,\Lambda} \varphi$ for $\rho_2(\varphi)$. By Definition~\ref{correlationfunctions}, we have:
$$ \int \rho_{2,\Lambda} \varphi = \Esp_{P_{\Lambda}} \left[ \lc \varphi, \cdot \rc \right].$$

Let $u_{p,i}= a_{i+p}-a_{i}$, with the convention $a_{N+l}=a_l+N$, and let us write the expectation $\Esp_{P_{\Lambda}} \left[ \lc \varphi, \cdot \rc \right]$ as:
\begin{multline} \label{Rho2Lambda}
\Esp_{P_{\Lambda}} \left[ \lc \varphi, \cdot \rc \right] = \int \Big( \sum_{x \neq y | x,y \in \mathcal{C}} \varphi(x,y) \Big) dP_{\Lambda}(\mathcal{C}) = 
\frac{1}{N} \int_{a_1}^{a_1 + N} \Big(\sum_{x \neq y | x,y \in \theta_{t} \cdot \Lambda} \varphi(x,y)\Big) dt \\
= \frac{1}{N} \sum_{i=1}^N \int_{a_i}^{a_{i+1}}\Big(  \sum_{k \in \mz} \sum_{p=1}^{+\infty} \left(\varphi(a_{i+k} - t, a_{i+k} - t + u_{p,i+k}) + \varphi(a_{i+k} -t + u_{p,i+k}, a_{i+k} -t)\right)\Big) dt.
\end{multline}
The first equality in (\ref{Rho2Lambda}) is simply an explicitation of the measure $P_{\Lambda}$ as an average of $\Lambda$ on translations in any interval of length $N$ as in (\ref{PLambda}), and the second equality amounts to writing the sum of $\varphi$ over couples of distinct points by taking $a_i$ as the “origin” of $\Lambda$ on each interval $[a_i, a_{i+1}]$:
\begin{multline*}
\sum_{x \neq y | x,y \in \theta_{t} \cdot \Lambda} \varphi(x,y) = \sum_{x < y | x,y \in \theta_{t} \cdot \Lambda} \varphi(x,y) + \varphi(y,x) = \sum_{x < y | x,y \in  \Lambda} \varphi(x-t,y-t) + \varphi(y-t,x-t) \\ = \sum_{k \in \mz} \sum_{p=1}^{+\infty} \varphi(a_k -t, a_{k+p} -t) + \varphi(a_{k+p} -t,a_k -t) \\ = \sum_{k \in \mz} \sum_{p=1}^{+\infty} \varphi(a_{i+k} -t, a_{i+ k+p} -t) + \varphi(a_{i+ k+p} -t,a_{i+k} -t)
\end{multline*}
and using the fact that, by definition, $a_{i+ k+p} = a_{i+k} + u_{p, i+k}$.

For $i = 1 \dots N$, let $c_i = 1$ when $|u_{1,i} - 1| \leq 1$ and $c_i =  \floor*{u_{1,i}}$ otherwise, so that $|u_{1,i} - c_i|  \leq \min\left(|b_{1,i}|, 1\right)$. Let us recall that the numbers $b_{p,i}$ are defined as $b_{p,i} =u_{p,i} -p$. The following bounds are easily seen
\begin{equation}\label{borneChasles}
\sum_{i=1}^N |u_{1,i} - c_i| \leq \sum_{i=1}^N \min\left(|b_{1,i}|, 1\right) \text{ and }\left|N- \sum_{i=1}^N c_i\right| \leq \sum_{i=1}^N \min\left(|b_{1,i}|, 1\right),
\end{equation}
the second inequality following from the first one by observing that since $\sum_{i=1} u_{1,i} = N$ we have $N- \sum_{i=1}^N c_i = \sum_{i=1}^N \left(u_{1,i} - c_i\right)$. We may now write that, by $1$-periodicity of $\mz$, and the fact that $c_i$ ($i = 1 \dots N$) is an integer: 
\begin{equation}\label{decompZ} \int_{0}^{1} \delta_{\theta_t \cdot \mathbb{Z}} dt = \frac{1}{N} \sum_{i=1}^N \int_{0}^1 \left(\int_{0}^{u_{1,i}}\delta_{\theta_t \cdot \mathbb{Z}} dt + \int_{u_{1,i}}^{c_i} \delta_{\theta_t \cdot \mathbb{Z}} dt\right) + \frac{1}{N} \int_{0}^{N-\sum_{i=1}^N c_i} \delta_{\theta_t \cdot \mathbb{Z}} dt.
\end{equation}

The decomposition of \eqref{decompZ} is meant to adapt the average of $\mz$ over translations in $[0,N]$ to the decomposition as a sum over translations in $[a_i, a_{i+1}]$ used in \eqref{Rho2Lambda}, at the cost of an error term which feels the spacing irregularities in $\Lambda$. Using \eqref{decompZ} when testing against $\varphi$ yields, by making a change of variables $t \mapsto t+a_i$ on each interval $[0, u_{1,i}]$:
\begin{multline*}
\int \rho_{2,P_{\mz}} \varphi = \frac{1}{N}  \sum_{i=1}^N \int_{a_i}^{a_{i+1}} \left(\sum_{k\in \mz} \sum_{p=1}^{+ \infty}  \varphi(a_i + k-t, a_i + k- t + p) + \varphi(a_i+ k - t +p,a_i + k -t)\right)dt \\ + \frac{1}{N} \sum_{i=1}^N \int_{u_{1,i}}^{c_i} \lc \varphi, \theta_{t} \cdot \mz \rc dt + \frac{1}{N} \int_{0}^{N-\sum_{i=1}^N c_i} \lc \varphi, \theta_{t} \cdot \mz \rc dt,
\end{multline*} 
where we have used the same way of writing $\Esp_{P_{\mz}} \left[ \lc \varphi, \cdot \rc\right]$ as in \eqref{Rho2Lambda}. Since $\varphi$ is compactly supported on $[-T, T]^2$, the terms $\lc \varphi, \theta_{t} \cdot \mz \rc$ are bounded uniformly on $t \in \mr$ by $(2T+1)^2 ||\varphi||_{\infty}$, because there is at most $(2T+1)^2$ couples of distinct points of $\mz$ in any interval of length $2T$. Since we may bound the lengths of the intervals $|u_{1,i}-c_i|$ ($i= 1 \dots N$) and $|N-\sum_{i=1}^N c_i|$ according to (\ref{borneChasles}), we get
\begin{multline}\label{Rho2Z}
\int \rho_{2,P_{\mz}} \varphi = \frac{1}{N} \sum_{i=1}^N \int_{a_i}^{a_{i+1}} \Big(\sum_{k\in \mz} \sum_{p=1}^{+ \infty}  \varphi(a_i - t + k, a_i- t + k + p) \\ + \varphi(a_i - t +k +p,a_i -t + k) \Big) dt + \frac{1}{N} \sum_{i=1}^N \min \left(|b_i|, 1\right) O(||\varphi||_{\infty})
\end{multline}
where the terms $O(||\varphi||_{\infty})$ are bounded by $(2T+1)^2 ||\varphi||_{\infty}$. In the rest of the proof, we denote by $C_{\varphi}$ a constant, which may vary from line to line, depending only on $\varphi$ via $||\varphi||_{\infty}$ and $||\nabla \varphi||_{\infty}$ and $T$.

Let us recall that $a_{i+k} = a_i + u_{k,i}$. A first order expansion of $\varphi$ yields, for any $t,i,k,p$ such that $a_{i+k}$ and $a_{i+k+p}$ lie in $[-T +t, T+t]$, 
\begin{multline} \label{devphi}
\left| \varphi(a_{i+k} -t, a_{i+k} -t + u_{i,k+p}) - \varphi(a_i - t + k, a_i- t + k + p) \right| \\ \leq \min \left(||\nabla \varphi||_{\infty} |b_{k,i}|, 2 ||\varphi||_{\infty}\right) + \min \left(||\nabla \varphi||_{\infty} |b_{k+p,i}|, 2 ||\varphi||_{\infty}\right)  \\ \leq C_{\varphi} \left (\min(|b_{k,i}|, 1) + \min(|b_{k+p,i}|,1) \right).
\end{multline}

We may now compare (\ref{Rho2Lambda}) and the main term of (\ref{Rho2Z}) by summing the expansions (\ref{devphi}):
\begin{multline} \label{Deltarho1}
\Bigg|\frac{1}{N} \sum_{i=1}^N \int_{a_i}^{a_{i+1}}\bigg(  \sum_{k \in \mz} \sum_{p=1}^{+\infty} \varphi(a_{i+k} - t, a_{i+k} - t + u_{p,i+k}) + \varphi(a_{i+k} -t + u_{p,i+k}, a_{i+k} -t)\\
-  \varphi(a_i-t+k, a_i-t+k+p) - \varphi(a_i-t+k+p,a_i -t + k) \bigg) dt\Bigg| \\
\leq C_{\varphi} \frac{1}{N} \sum_{i=1}^N \sum_{k=1}^{\infty} m_{k,i} \min(|b_{k,i}|,1)  \end{multline}
where the numbers $m_{k,i}$ are given by
$$m_{k,i} = \int_{a_1}^{a_{N+1}} \mathbf{1}_{a_i \in [-T + t, T+t]} \mathbf{1}_{a_{i+k} \in [-T +t, T+t]} dt.$$
Indeed, the first-order expansions (\ref{devphi}) allow us to bound every term in the left-hand side of (\ref{Deltarho1}) by a sum of four terms of the type $\min(|b_{k,i}|,1)$. For $t \in [a_1, a_{N+1}]$ the term $\min(|b_{k,i}|,1)$ appears only if $a_{i}-t$ and $a_{i+k}-t$ lie in $[-T, T]$, which gives the expression for $m_{k,i}$. Moreover since there is at most $N$ points of $\Lambda$ in any interval of length $N$, if $N \geq 2T$ we have $m_{i,k} = 0$ for all $k > N$. The assumption $N \geq 2T$ is not restrictive since we may always consider a $N$-periodic configuration as $rN$-periodic for any integer $r$.

It is easy to see that $m_{k,i} \leq 2T$, and if $m_{k,i}$ is nonzero it means that $u_{k,i} \leq 2T$ (since $a_i$ and $a_{i+k}$ lie in some common interval of length $2T$) hence the spacing error $|b_{k,i}|$ is larger than $k - 2T$. Consequently, if $m_{k,i}$ is nonzero for $k \geq 3T$, we have $\frac{b_{k,i}}{k} \geq \frac{1}{3}$ so that for any $i, k$, since $T \geq 1$:
\begin{equation}\label{mikbon}
m_{i,k} \min(|b_{k,i}|,1) \leq 3T m_{i,k} \min\left(\frac{|b_{k,i}|}{k},1\right).
\end{equation}
%\begin{multline} \label{splittingmik}
%\sum_{i=1}^N \sum_{k=1}^{\infty} m_{i,k} \min(|b_{k,i}|,1) \leq 2T \left(\sum_{i=1}^N \sum_{k=1}^{3T}  m_{i,k} \min(|b_{k,i}|,1) + \sum_{i=1}^N \sum_{k=3T}^{\infty}  m_{i,k} \mathbf{1}_{ \frac{b_{k,i}}{k} \geq \frac{1}{3}} \min(|b_{k,i}|,1)\right) \\
%\leq 2T\left( \sum_{i=1}^N \sum_{k=1}^{3T} 3T m_{i,k} \min\left(\frac{|b_{k,i}|}{k},1\right) + \sum_{i=1}^N \sum_{k=3T}^{\infty}  3 m_{i,k} \min\left(\frac{|b_{k,i}|}{k},1\right)\right) \\ \leq 6T^2 \sum_{i=1}^N \sum_{k=1}^{\infty} m_{i,k} \min\left(\frac{|b_{k,i}|}{k},1\right).
%\end{multline}
Now we may bound $\sum_{i=1}^N\sum_{k=1}^{N} m_{i,k}^2$ the following way :
\begin{multline*}
\sum_{i=1}^N\sum_{k=1}^{N} m_{i,k}^2 \leq 2T \sum_{i=1}^N\sum_{k=1}^{N} m_{i,k} = 2T  \int_{a_1}^{a_{N+1}}  \left(\sum_{i=1}^N\sum_{k=1}^{N}  \mathbf{1}_{a_i \in [-T + t, T+t]} \mathbf{1}_{a_{i+k} \in [-T +t, T+t]}\right) dt \\
\leq 2T \int_{a_1}^{a_{N+1}} \N^2(\theta_{t}\cdot \Lambda, T)	
\end{multline*}
where $\N(\Lambda, T)$ denotes the number of points of $\Lambda$ in $[-T, T]$. By definition of $P_{\Lambda}$ we may re-write the last term as 
$$\int_{a_1}^{a_{N+1}} \N^2(\theta_{t}\cdot \Lambda, T) = N \int \N^2(\mathcal{C}, T) dP_{\Lambda}(\mathcal{C}).$$
By Lemma \ref{discrepa} we know that $$\int \N^2(\mathcal{C}, T) dP_{\Lambda}(\mathcal{C}) \leq  C_T(C + \bw(P_{\Lambda})) =  C_T(C + W(E_{\{a_i\}})),$$ so that we finally get:
\begin{equation}\label{splittingmik}
\sum_{i=1}^N\sum_{k=1}^{N} m_{i,k}^2 \leq N C_T(C + W(E_{\{a_i\}})).
\end{equation}
Combining (\ref{Rho2Z}) and (\ref{Deltarho1}) we obtain:
$$\left| \int \rho_{2,P_{\mz}} - \rho_{2, P_{\Lambda}} \varphi \right| \leq \frac{1}{N} \sum_{i=1}^N \sum_{k=1}^{N}  3T m_{i,k}   \min\left(\frac{|b_{k,i}|}{k},1\right) +  \frac{1}{N} \sum_{i=1}^N \min \left(|b_{1,i}|, 1\right) O(||\varphi||_{\infty}).$$
The second sum is bounded by Lemma \ref{QLB} as follows:
$$\frac{1}{N} \sum_{i=1}^N \min \left(|b_{1,i}|, 1\right) \leq \left(\frac{1}{N} \sum_{i=1}^N \min \left(|b^2_{1,i}|, 1\right)\right)^{\hal} \leq C \left(W(E_{\{a_i\}})- W(\mz)\right)^{\hal}$$
hence we have
\begin{equation} \label{Deltarhom}
\left| \int \rho_{2,P_{\mz}} - \rho_{2, P_{\Lambda}} \varphi \right| \leq C_{\varphi} \frac{1}{N} \sum_{i=1}^N \sum_{k=1}^{N} m_{i,k} \min\left(\frac{|b_{k,i}|}{k},1\right) + C \left(W(E_{\{a_i\}})- W(\mz)\right)^{\hal}.
\end{equation}

Using the Cauchy-Schwarz inequality in (\ref{Deltarhom}), the bound \eqref{splittingmik} on the $m_{i,k}$ and the bound (\ref{eqQLB}) of Lemma \ref{QLB} we get:
\begin{multline}
\left| \int \left(\rho_{2,P_{\mz}} - \rho_{2, P_{\Lambda}}\right) \varphi \right| \leq C_{\varphi} \left(\frac{1}{N} \sum_{i=1}^N \sum_{k=1}^{N} m_{i,k}^2\right)^{\hal} \left(\frac{1}{N}\sum_{i=1}^N \sum_{k=1}^{N} \min\left(\frac{|b_{k,i}|^2}{k^2},1\right)\right)^{\hal}  \\ + C \left(W(E_{\{a_i\}}- W(\mz)\right)^{\hal} \leq C_{\varphi} \left(C + W(\j_{\{a_i\}})\right)^{\hal} \left(W(\j_{\{a_i\}}) - W(\mz)\right)^{\hal} 
\end{multline}
which concludes the proof of the lemma. Let us note that altough the bound of Lemma~\ref{QLB} only controls $\sum_{i=1}^N \sum_{k=1}^{N/2} \min\left(\frac{|b_{k,i}|^2}{k^2},1\right)$ we may easily bound $\sum_{i=1}^N \sum_{k=1}^{N} \min\left(\frac{|b_{k,i}|^2}{k^2},1\right)$ as well by periodicity.
\end{proof}

\subsection{Extension to the non-periodic case}
\label{preuveTheoreme}

Let us now turn to the proof of the main result, Theorem \ref{NonPer}.
\begin{proof} In the following we denote by $C_{\varphi}$ a constant, which may vary from line to line, depending only on $\varphi$ via $||\varphi||_{\infty}$ and $||\nabla \varphi||_{\infty}$ and $T$.

Since $\bw(P)$ is finite, let us recall that by Remark \ref{meandensitym} the two-point correlation function of $P_{\Lambda}$ exists at least in distributional sense.

\noindent
{\it -  Step 1: Choosing a large set where the controls are uniform}.
A straightforward adaptation of \cite[Lemma 3.6.]{SSlg} (the only modification is that we are dealing with probability measures on the electric fields only, with no dependance on $\Sigma$) ensures that for any $\epsilon > 0$, we may find a subset $G_{\epsilon} \subset \mathcal{A}_1$ such that $G_{\epsilon}$ has almost full $P$-measure, and on which we have a uniform control for the relevant quantities. Precisely, the lemma ensures that:
\begin{enumerate}
\item $P({G_\ep}^c) <\epsilon$
\item The convergence \eqref{Wroi} in the definition of the renormalized energy is uniform with respect to $\j\in G_\ep$.
\item Writing $\div \j = 2\pi(\nu_\j-1)$, both $W(\j)$ and $\nu_\j(I_R)/R$ are bounded uniformly with respect to $\j\in G_\ep$ and $R>1$.
\item Uniformly with respect to $\j\in G_\ep$ we have
\begin{equation}
\lim_{y_0 \rightarrow +\infty} \lim_{R \rightarrow +\infty} \fint_{I_R} \int_{|y|>y_0} |\j|^2 =0.
\end{equation}
This is a technical assumption needed for the “screening” construction of Step 2.
\end{enumerate}
Moreover, we may assume (this is Equation (5.3) in \cite[Lemma 3.6.v)]{SSlg}) that $G_{\epsilon}$ is almost translation-invariant in that for any $E \in G_{\epsilon}$, $E(\lambda + \cdot) \in G_{\epsilon}$ for all $\lambda \in \mr$ except for a set of bounded Lebesgue measure (the set depends on $E$ but its measure is bounded uniformly on $G_{\epsilon}$). Note that, strictly speaking, it is not precised in \cite[Lemma 3.6.]{SSlg} that one may choose $G_{\epsilon}$ both of almost full $P$-measure and almost translation-invariant, however it is a consequence of Equation (3.6.) in  \cite[Lemma 3.6.v)]{SSlg}), and is written as Equation (7.6) in \cite[Lemma 7.6]{SS2d} (which handles the purely 2D case, but from which \cite[Lemma 3.6.]{SSlg} is essentially deduced).

For $\epsilon < 1$, let $P_{\epsilon}$ be the probability measure induced by $P$ on $G_{\epsilon}$, let $P_{\Lambda, \epsilon}$ be the push-forward of $P_{\epsilon}$ by the map $E \mapsto \frac{1}{2\pi} \div E + \delta_{\mr}$ and let $\rho_{2, P_{\Lambda, \epsilon}}$ be the two-point correlation function of $P_{\Lambda, \epsilon}$. In the rest of the proof we make the following abuse of notation: we denote by $\mathbf{1}_{G_{\epsilon}}$ both the characteristic function of $G_{\epsilon}$ and its push-forward by the map $E \mapsto \frac{1}{2\pi} \div E + \delta_{\mr}$. We claim that
\begin{equation} \label{rhoepsilon}
\left | \int (\rho_{2,{P_{\Lambda}}} - \rho_{2, P_{\Lambda,\epsilon}})  \varphi \right | = o_{\epsilon \to 0}(1).
\end{equation}
Indeed, we know that $\int \rho_{2,P_{\Lambda}} |\varphi| = \Esp_{P_{\Lambda}} \left[ \lc |\varphi|, \cdot \rc \right] $ is finite (see Lemma \ref{boundphiW}), and that $P(G_{\epsilon}^c) < \epsilon$. By uniform continuity of the integral, if $\epsilon$ is small enough, then $\left| \Esp_{P_{\Lambda}} \left[ \lc \varphi, \cdot \rc \right]  - \Esp_{P_{\Lambda}} \left[ \mathbf{1}_{G_{\epsilon}} \lc \varphi, \cdot \rc \right] \right|$ is arbitrarily small. This proves the claim, because we also have, by definition of $P_{\Lambda,\epsilon}$:
$$\int \rho_{2, P_{\Lambda,\epsilon}}  \varphi = \Esp \left[ \mathbf{1}_{G_{\epsilon}}   \lc \varphi, \cdot \rc \right] \frac{1}{P(G_{\epsilon})}.$$

\noindent
{\it -  Step 2: Obtaining periodic fields by screening}.
We now construct, for $R$ large enough and for each $E$ in $G_{\epsilon}$, a periodic field $\bar{E}_R$ of period $R$, which approximates $E$, and we use these fields to approximate $P_{\epsilon}$ by an average of stationary measures on $\mathcal{A}_1$ associated to periodic electric fields.

To this aim, we apply \cite[Proposition 3.1.]{SSlg}. This screening result allows us to truncate $E$ outside of a large interval, to approximate $E$ on this interval by some field which is “screened” so that we may paste identical copies of it in order to get a periodic electric field on $\mr$, whilst letting $E$ unchanged in some large interval. For $R > 0$ we let $I_R = [-R/2, R/2]$.

Let $\alpha > 0$. We get from \cite[Proposition 3.1.]{SSlg} that there exists $R_0 >0$ (depending on $\epsilon$ and $\alpha$) such that for every integer $R \geq R_0$, for every $\j \in G_{\epsilon}$, there exists a vector field $\j_R\in L^q_{loc}(I_R \times \mr, \mr^2)$ (for $q <2$)
satisfying:
\begin{itemize}
\item[i)]   $\j_R \cdot \vec{\nu} =0$ on $\p I_R \times \mr$, where $\vec{\nu}$ denotes the outer unit normal.
\item[ii)] There is a discrete subset $\Lambda \subset I_R$ such that $$\div \j_R = 2\pi \(\sum_{p\in\Lambda}\delta_p -  \delta_\mr\) \quad \text{in} \ I_R \times \mr.$$
\item[iii)]  $\j_R(z) = \j(z)$ for $x \in [- R/2+\alpha R, R/2-\alpha R]$.
\item[iv)]
\begin{equation}\label{restronque} \frac{W(\j_R,\indic_{I_R })}{R} \le W(\j)+ \alpha.\end{equation}
\end{itemize}
The “screened” property is expressed by i), the point iii) shows that $E$ is unchanged on a large interval and iv) gives an upper bound on the new energy.

For any integer $R \geq R_0$, we extend the electric fields $E_{R}$ periodically, and make them gradients. This amounts to first pasting together identical copies of $E_{R}$ to make it periodic of period $R$ (the point i) allows us to make such a construction), and then considering the $L^2$-projection of the constructed field onto the space of gradient vector fields, which, together with point ii) guarantees that we end up in the class $\mathcal{A}_1$. It is proved that the projection can only decrease the energy, so that iv) is conserved. Moreover, projecting onto gradients leave the divergence of $E_{R}$ unchanged, so that iii) becomes:
$$\div \bar{E}_R(z) = \div E(z) \text{ for } x \in [- R/2+\alpha R, R/2-\alpha R].$$
Details are given in the proof of \cite[Proposition 4.1.]{SSlg}, and we only state the conclusions: we get, for each $E \in G_{\epsilon}$, and any $R \geq R_0$ (let us emphasize that $R_0$ depends on $\epsilon$ and $\alpha$) an electric field $\bar{E}_{R}$ which is $R$-periodic, which coincides with $E$ on $[- R/2+\alpha R, R/2-\alpha R]$, and such that
$$\frac{W(\bar{E}_R,\indic_{I_R })}{R} \le W(\j)+ \alpha.$$ 

\noindent
{\it -  Step 3: Approximate stationary processes}.
For each $E \in G_{\epsilon}$, and any $R \geq R_0$, we now consider the stationary probability measure $\fint_{-R/2}^{R/2} \delta_{\theta_t \cdot \bar{E}_R} dt$ on $\mathcal{A}_1$ associated to $\bar{E}_R$ by averaging $\bar{E}_R$ over translations in $[-R/2, R/2]$, and we define $P^{R}_{\epsilon}$ as the pushforward of the probability measure $P_{\epsilon}$ by the map
$$E \mapsto \fint_{-R/2}^{R/2} \delta_{\theta_t \cdot \bar{E}_R} dt$$
(let us note that this map is only defined on $G_{\epsilon}$, but $P_{\epsilon}$ itself is concentrated on $G_{\epsilon}$). The process $P^{R}_{\epsilon}$ is stationary as an average of stationary probability measures, we denote by $P^{R}_{\Lambda, \epsilon}$ its push-forward by the map $E \mapsto \frac{1}{2\pi} \div E + \delta_{\mr}$ and we let $\rho_{2, P_{\Lambda,\epsilon}}^{R}$ be the two-point correlation function of $P^{R}_{\Lambda, \epsilon}$. We now claim that
\begin{equation}
\label{rhoper}
\left|\int (\rho_{2, P_{\Lambda,\epsilon}}^{R} - \rho_{2, P_{\Lambda,\epsilon}}) \varphi\right| \leq  o_{R \to \infty}(1) + \alpha C_{\varphi} (\bar{W}(P) +C).
\end{equation}

Indeed, by definition we have
\begin{multline}\label{ro2pep} \int \rho_{2, P_{\Lambda,\epsilon}}^{R} \varphi = \Esp_{P^R_{\Lambda, \epsilon}} \left[ \fint_{-R/2}^{R/2} \lc \varphi, \theta_t  \cdot \rc   dt \right] = \frac{1}{R} \int_{-R/2 + 2\alpha R}^{R/2 - 2\alpha R} \Esp_{P^R_{\Lambda,\epsilon}} \left[ \lc \varphi, \theta_t  \cdot \rc  \right]  dt \\ + \frac{1}{R} \int_{-R/2}^{-R/2 + 2\alpha R} \Esp_{P^R_{\Lambda,\epsilon}} \left[ \lc \varphi, \theta_t  \cdot \rc  \right] dt+ \frac{1}{R} \int_{R/2-2\alpha R}^{R/2} \Esp_{P^R_{\Lambda,\epsilon}} \left[ \lc \varphi, \theta_t  \cdot \rc  \right] dt.
\end{multline}
Since $\varphi$ is compactly supported and since $\div \bar{E}_R$ coincides with $\div E$ on the interval $[- R/2+ \alpha R, R/2- \alpha R]$, if $R$ is large enough (depending on $\varphi$) we have $P_{\epsilon}$-a.s. that $\div \bar{E}_R(\cdot -t) = \div E(\cdot -t)$ for $t \in [- R/2+2 \alpha R, R/2- 2 \alpha R]$ (i.e. the screening and periodization have not affected the point configuration on a large interval). It means that for $R$ large enough, we may express the first integrand in the right-hand side of (\ref{ro2pep}) as
$$ \Esp_{P^R_{\Lambda, \epsilon}} \left[ \lc \varphi, \theta_t  \cdot \rc  \right] =  \Esp_{P_{\Lambda, \epsilon}} \left[ \lc \varphi, \theta_t \cdot \rc  \right] \\ = \frac{1}{P(G_{\epsilon})} \Esp_{P_{\Lambda}} \left[ \mathbf{1}_{G_{\epsilon}} \lc \varphi, \theta_t \cdot \rc \right].$$

The probability measure $P$ is, by assumption, translation-invariant hence so is $P_{\Lambda}$, so that for any $t \in \mr$ we have
$$ \Esp_{P_{\Lambda}} \left[\mathbf{1}_{G_{\epsilon}} \lc \varphi, \theta_t  \cdot \rc \right] = \Esp_{P_{\Lambda}} \left[\mathbf{1}_{G_{\epsilon}}(\theta_{-t} \cdot) \lc \varphi, \cdot \rc \right],
$$ 
which in turn gives
$$\Esp_{P_{\Lambda, \epsilon}} \left[ \lc \varphi, \theta_t  \cdot \rc  \right] = \frac{1}{P(G_{\epsilon})}  \Esp_{P_{\Lambda}} \left[\mathbf{1}_{G_{\epsilon}} \lc \varphi, \theta_t  \cdot \rc \right] = \frac{1}{P(G_{\epsilon})} \Esp_{P_{\Lambda}} \left[\mathbf{1}_{G_{\epsilon}}(\theta_{-t} \cdot) \lc \varphi, \cdot \rc  \right].$$

We now claim to control the default of invariance of $P_{\Lambda,\epsilon}$ under translations the following way: 
\begin{equation} \label{invaPeps}
\left|\frac{1}{R} \int_{-R/2 + 2\alpha R}^{R/2 - 2\alpha R} dt \Esp_{P_{\Lambda,\epsilon}} \left[ \lc \varphi, \theta_t \cdot \rc  \right] - \frac{R-4\alpha R}{R} \Esp_{P_{\Lambda,\epsilon}} \left[ \lc \varphi, \cdot \rc \right] \right| = o_{R \to \infty}(1)
\end{equation}
with a $o_{R \to \infty}(1)$ depending on $\varphi, \epsilon, P$.

Indeed, we have, for $t \in  [- R/2+2 \alpha R, R/2- 2\alpha R]$:
\begin{equation} \label{deltaPeps}
 \Esp_{P_{\Lambda, \epsilon}} \left[ \lc \varphi, \theta_t  \cdot \rc  \right] - \Esp_{P_{\Lambda,\epsilon}} \left[ \lc \varphi, \cdot \rc  \right]  = \frac{1}{P(G_{\epsilon})} \int \lc \varphi, \mathcal{C} \rc \left(\mathbf{1}_{G_{\epsilon}}(\theta_{-t} \cdot \mathcal{C}) - \mathbf{1}_{G_{\epsilon}}(\mathcal{C}) \right) dP_{\Lambda}(\mathcal{C}).
\end{equation}
Integrating (\ref{deltaPeps}) between $ [- R/2+2 \alpha R, R/2- 2 \alpha R]$ yields: 
\begin{multline}
\frac{1}{R} \int_{-R/2 + 2 \alpha R}^{R/2 - 2 \alpha R} dt \Esp_{P_{\Lambda, \epsilon}} \left[ \lc \varphi, \theta_t  \cdot \rc  \right] - \frac{R-4\alpha R}{R} \Esp_{P_{\Lambda,\epsilon}} \left[ \lc \varphi, \cdot \rc  \right]  \\ = \frac{1}{R} \int_{-R/2 + 2 \alpha R}^{R/2 - 2 \alpha R} dt \frac{1}{P(G_{\epsilon})} \int  \lc \varphi, \mathcal{C} \rc \left(\mathbf{1}_{G_{\epsilon}}(\theta_{-t} \cdot \mathcal{C}) - \mathbf{1}_{G_{\epsilon}}(\mathcal{C})\right) dP_{\Lambda}(\mathcal{C}) \\ 
= \frac{1}{RP(G_{\epsilon})} \int dP_{\Lambda}(\mathcal{C})  \lc \varphi, \mathcal{C} \rc \int_{-R/2 + 2 \alpha R}^{R/2 - 2 \alpha R} dt  \left(\mathbf{1}_{G_{\epsilon}}(\theta_{-t}  \cdot \mathcal{C}) - \mathbf{1}_{G_{\epsilon}}(\mathcal{C})\right).
\end{multline} 

We know that for $E \in G_{\epsilon}$, there is a set $\Gamma(E)$ such that $|\Gamma(E)| \leq C_{\epsilon}$ (for some constant depending only on $G_{\epsilon}$) and if $\lambda \notin \Gamma(E)$ then $E(\cdot - \lambda) \in G_{\epsilon}$. This property is clearly pushed forward at the level of the point configurations.
This yields the following bound 
$$\left|\int_{-R/2 + 2 \alpha R}^{R/2 - 2 \alpha R} dt  \left(\mathbf{1}_{G_{\epsilon}}(\theta_{-t} \cdot \mathcal{C}) - \mathbf{1}_{G_{\epsilon}}(\mathcal{C})\right)\right| \leq C_{\epsilon}$$
and since $\int dP_{\Lambda}(\mathcal{C}) \lc |\varphi|, \mathcal{C} \rc$ is finite (again, by Lemma \ref{boundphiW}) we get
\begin{equation} \label{errroeps1}
\left|\frac{1}{R} \int_{-R/2 + 2\alpha R}^{R/2 - 2\alpha R} dt \Esp_{P_{\Lambda,\epsilon}} \left[ \lc \varphi, \theta_t  \cdot \rc  \right] - \frac{R-4\alpha R}{R} \Esp_{P_{\Lambda,\epsilon}} \left[ \lc \varphi, \cdot \rc  \right]\right| \leq \frac{C_{\epsilon, \varphi, P}}{R}
\end{equation}
with a constant depending on $\epsilon, \varphi, P$, which proves (\ref{invaPeps}).

We are now left to bound the two error terms in (\ref{ro2pep}), for which we have, applying Lemma \ref{boundphiW} in the last inequality:
\begin{equation} \label{errroeps2}
\left| \frac{1}{R} \int_{R/2-2\alpha R}^{R/2} \Esp_{P^R_{\Lambda,\epsilon}} \left[ \lc \varphi, \theta_t  \cdot \rc  \right] \right| dt \leq 2\alpha \sup_{t \in \mr} \Esp_{P^R_{\Lambda, \epsilon}} \left[ \lc |\varphi|, \theta_t  \cdot \rc  \right] \leq \alpha C_{\varphi} (\bw(P) + C).
\end{equation}
The other term $\frac{1}{R} \int_{R/2}^{R/2-2\alpha R} \Esp_{P^R_{\Lambda,\epsilon}} \left[ \lc \varphi, \theta_t  \cdot \rc  \right] dt$ is bounded the same way, moreover with the same application of Lemma \ref{boundphiW} we get
\begin{equation} \label{errroeps3}
\left|4\alpha \Esp_{P^R_{\Lambda, \epsilon}} \left[ \lc \varphi, \cdot \rc  \right]\right| \leq \alpha C_{\varphi} (\bw(P) + C).
\end{equation}

Observing that by definition $\Esp_{P_{\Lambda, \epsilon}} \left[ \lc \varphi, \cdot \rc  \right]  = \int \rho_{2, P_{\Lambda, \epsilon}}  \varphi$, and combining (\ref{ro2pep}) with the estimates (\ref{errroeps1}), (\ref{errroeps2}), (\ref{errroeps3}), we have 
\begin{equation}\label{erroeps4}
\left|\int (\rho_{2, P_{\Lambda, \epsilon}}^{R} - \rho_{2, P_{\Lambda, \epsilon}}) \varphi\right| \leq \frac{C_{\epsilon, \varphi, P}}{R} +  \alpha C_{\varphi} (\bw(P) +C)
\end{equation}
which proves the claim (\ref{rhoper}).

\noindent
{\it -  Step 4: Using the result of the periodic case}.
We may now come back to the proof of Theorem \ref{NonPer}. Let us fix $\eta > 0$, and take $\alpha = \frac{\eta}{C_{\varphi} (\bw(P) +C)}$, where $C_{\varphi}$ and $C$ are the constant in (\ref{erroeps4}). Then for $R$ large enough (depending on $\alpha$ and $G_{\epsilon}$) we have
\begin{equation} \label{comparrhoz}
\left|\int (\rho_{2, P_{\Lambda,\epsilon}}^{R} - \rho_{2, P_{\Lambda,\epsilon}}) \varphi\right| \leq \eta + \frac{C_{\varphi, \epsilon, P}}{R}.
\end{equation}

Let us now apply Lemma \ref{LemmeDeltarho} for the periodic case. For each $E$ (under $P_{\epsilon}$), and for any $R > 0$, we consider the stationary measure $ \fint_{-R/2}^{R/2} \delta_{\theta \cdot \bar{E}_R} dt$ whose energy is finite $P_{\epsilon}$-a.s., and we denote by $\rho_{2, E,R}$ the two-point correlation function of its push-forward by the map (\ref{pushforward}). From Proposition \ref{NonPer} we get
$$\left| \int_{\mr^2} (\rho_{2, E,_R}  -\rho_{2, P_{\mz}}) \varphi\right| \leq  C_{\varphi} (C+W(\bar{E}_R))^{\hal} \left(W(\bar{E}_R) - W(\mz)\right)^{1/2} $$
and integrating this inequality against $dP_{\epsilon}(E)$ gives (using Jensen's inequality in the last line)
\begin{multline} \label{comparrhoa}
\left|\int_{\mr^2} (\rho_{2, P_{\Lambda,\epsilon}}^{R} -\rho_{2, P_{\mz}}) \varphi \right| = \left| \int dP_{\epsilon}(E) \int_{\mr^2} (\rho_{2, E,R}  -\rho_{2, P_{\mz}}) \varphi\right| \\ \leq  C_{\varphi} (C+\bw(P^R_{\epsilon}))^{\hal} \left(\bw(P^R_{\epsilon}) - W(\mz)\right)^{1/2}. \end{multline}
By construction we know that for $R$ large enough (depending on $G_{\epsilon}$ and $\alpha$) we have $P_{\epsilon}$-a.s.
$$W(\bar{E}_R) \leq W(E) + \eta $$
hence (\ref{comparrhoa}) gives, for $R$ large enough
\begin{equation} \label{comparrhob}
\left|\int_{\mr^2} (\rho_{2, P_{\Lambda, \epsilon}}^{R} -\rho_{2, P_{\mz}}) \varphi \right| \leq C_{\varphi} (C+\bw(P_{\epsilon})+\eta)^{\hal} \left(\bw(P_{\epsilon}) + \eta - W(\mz)\right)^{1/2}.
\end{equation}
Combinining (\ref{comparrhoz}) and (\ref{comparrhob}), we get
\begin{equation}  \label{comparrhoc}
\left|\int (\rho_{2, P_{\Lambda,\epsilon}} - \rho_{2, P_{\mz}}) \varphi\right| \leq \eta + \frac{C_{\varphi, \epsilon, P}}{R} + C_{\varphi} (C+\bw(P_{\epsilon})+\eta)^{\hal} \left(\bw(P_{\epsilon}) + \eta - W(\mz)\right)^{1/2}.
\end{equation}
Since $\int |W(E)| dP(E)$ is finite (because $\bw(P)$ is finite and $W$ is bounded below on $\mathcal{A}_1$), and since $P(G_{\epsilon}^c) < \epsilon$, by the uniform continuity of the integral we know that \begin{equation} \label{comparbarW}
\bw(P_{\epsilon})  = \bw(P) + o_{\epsilon \to 0}(1).
\end{equation}
Combining (\ref{comparrhoc}), (\ref{comparbarW}) and (\ref{rhoepsilon}), sending $\alpha$ to 0, $\epsilon$ to 0, and then $R$ to $+ \infty$, we  conclude the proof.
\end{proof}
\subsection{Uniqueness results}
\label{uniqueness}
We now turn to the proof of the uniqueness results for minimizers as stated in Corollary \ref{corollaire1} and Corollary \ref{corollaire2}. First we observe that the invariance condition in the definition of admissible measures is equivalent to translation-invariance of the disintegration measures (for a definition see \cite[Section 5.3.]{GradientFlows}):
\begin{remark} \label{TlambdathenT} Let $P$ be an admissible probability measure on $\Sigma \times \mathcal{A}$, and let $\{P^x\}_{x \in \Sigma}$ be the disintegration measures of $P$ on $\mathcal{A}$ with respect to $\Sigma$. Since the first marginal of $P$ is the normalized Lebesgue measure on $\Sigma$ we have, by defintion of disintegration measures, for any continuous map $f \in L^1(dP)$:
$$\Esp_{P} \left[ f \right] = \fint_{\Sigma} dx \int f(x,E) dP^x(E).$$
For any smooth cut-off function $\chi$ on $\mr$ and any $\lambda \in \mr$ we have, by the invariance property of $P$:
$$\fint_{\Sigma} \chi(x) dx \int f(x,E) dP^x(E) = \fint_{\Sigma} \chi(x) dx \int f(x,E(\lambda + \cdot)) dP^x(E).$$
A standard approximation argument (by taking a sequence $\{\chi_n\}$ converging to a Dirac mass at $x_0$) shows that for any $x_0 \in \Sigma$, any $\lambda \in \mr$ we have $$\int f(x_0,E) dP^{x_0}(E) = \int f(x_0,E(\lambda + \cdot)) dP^{x_0}(E),$$ hence $P^{x_0}$ is translation-invariant for all $x_0 \in \Sigma$.

Conversely it is easy to see that if $\{P^{x}\}_{x \in \Sigma}$ is a measurable family of translation-invariant probability measures such that each $P^{x}$ is concentrated on $\mathcal{A}_{m_0(x)}$, then $\frac{dx_{|\Sigma}}{|\Sigma|} \otimes P^{x}$ is an admissible probability measure. In particular, $P_0$ as defined in \eqref{P0} is admissible.
\end{remark}

We now give the proof of Corollary \ref{corollaire1} and Corollary \ref{corollaire2}.
\begin{proof}
It is clear, from the crystallization result of section \ref{thmini}, that $P_{\frac{1}{m}\mz}$ (resp. $P_0$) is indeed \textit{a} minimizer of $\bw$ (resp. of $\tW$) on stationary measures $\probas(\mathcal{A}_m)$ (resp. on admissible probability measures). It remains to show the uniqueness.
By the scaling relation (\ref{scaling}), it is enough to show that $P_{\mz}$ is the unique minimizer of $\bw$ on $\probas(\mathcal{A}_1)$ to prove the first claim. If $P \in \probas(\mathcal{A}_1)$ is another minimizer we have $\bw(P) = W(\mz)$ hence by Theorem (\ref{NonPer}), if $P_{\Lambda}$ denotes the push-forward of $P$ by the map (\ref{pushforward}), we have
$$
\rho_{2, P_{\Lambda}} = \ro_{2, P_{\mz}}
$$
where $\rho_{2, P_{\Lambda}}$ is the two-point correlation function of $P_{\Lambda}$. Let us note that, in general, two point processes sharing the same two-point correlation function may be distinct (for conditions under which two point processes sharing \textit{all} their $k$-point correlation functions are equal, see \cite{lenard}), but here the rigidity of the lattice structure ensures that $P_{\Lambda} = P_{\mz}$. 

Testing $\rho_{2, P_{\Lambda}} = \rho_{2, P_{\mz}}$ against smooth approximations of $$(x,y) \mapsto \mathbf{1}_{x-y \in \mz^c}\mathbf{1}_{[-T,T]}(x)\mathbf{1}_{[-T,T]}(y) \text{ and } (x,y) \mapsto \mathbf{1}_{x-y \in \mz}\mathbf{1}_{[-T,T]}(x)\mathbf{1}_{[-T,T]}(y)$$ for any $T$, we get that $P_{\Lambda}$-almost surely, the configuration $\mathcal{C}$ is a translated copy of $\mz$. Indeed, testing against the first function shows that for any $T$ there is $P_{\Lambda}$-a.s. no couple of points $x,y \in \mathcal{C} \cap [-T,T]$ such that $x-y \notin \mz$, hence $\mathcal{C}$ is $P_{\Lambda}$-a.s. a subset of (a translated copy of) $\mz$. Moreover, testing against the second function shows that the average number of points in $[-T,T]$ coincides with that of $P_{\mz}$ for all $T$, since $P_{\Lambda}$ is stationary this ensures that in fact $P_{\Lambda} = P_{\mz}$. This proves the first claim of uniqueness. 

To prove the second claim, let $P \in \mathcal{P}(X)$ be a minimizer of $\tW$ on the set of admissible probability measures, and let us write its disintegration 
$P=  \frac{dx_{|\E}}{|\E|} \otimes P^x $ where $x$-a.e. in $\Sigma$, $P^x$ is a probability measure on $\mathcal{A}_{m_0(x)}$, and since $P$ is admissible we also know that $P^x$ itself is translation-invariant, see remark \ref{TlambdathenT}. Since $P$ minimizes $\tW$, the stationary probability measure $P^x$ minimizes $\bw$ over $\mathcal{P}(\mathcal{A}_{m_0(x)})$ for almost every $x \in \Sigma$. By the first claim, this means that $P^x = P_{\frac{1}{m_0(x)} \mz}$ for almost every $x \in \Sigma$, which in turn ensures that $$P = \frac{dx_{|\E}}{|\E|} \otimes  P_{\frac{t1}{m_0(x)} \mathbb{Z}} = P_0.$$
\end{proof}
\paragraph{Acknowledgements} The author would like to thank his PhD supervisor Sylvia Serfaty for suggesting the problem as well as for fruitful discussions.

\bibliographystyle{alpha}
\bibliography{appendixbib}
\end{document}